\documentclass[preprint,12pt]{elsarticle}

\usepackage{amsmath,amssymb,amsfonts,amsthm,mathtools,mathdots}
\usepackage{xcolor}
\usepackage{bm}

\newtheorem{theorem}{Theorem}
\newtheorem{remark}{Remark}

\newtheorem{lemma}{Lemma}
\newtheorem{definition}{Definition}

\newcommand{\cC}{\mathcal{C}}

\newcommand{\cS}{\mathcal{S}}

\newcommand{\bfc}{\bm{c}}

\newcommand{\bfh}{\bm{h}}

\newcommand{\bfs}{\bm{s}}

\newcommand{\bfx}{\bm{x}}

\newcommand{\bzero}{\bm{0}}

\renewcommand{\leq}{\leqslant}

\renewcommand{\geq}{\geqslant}

\newcommand{\F}{\mathbb{F}}
\newcommand{\C}{\mathbb{C}}

\usepackage{mathtools}

\DeclarePairedDelimiter\abs{\lvert}{\rvert}

\DeclarePairedDelimiter\ceil{\lceil}{\rceil}
\DeclarePairedDelimiter\floor{\lfloor}{\rfloor}
\DeclarePairedDelimiter\parenv{\lparen}{\rparen}

\DeclarePairedDelimiter\set{\{}{\}}

\DeclarePairedDelimiter\spn{\langle}{\rangle}

\newcommand{\eqdef}{\triangleq}
\newcommand{\T}{\intercal}
\DeclareMathOperator{\wt}{wt}
\DeclareMathOperator{\supp}{supp}
\newcommand{\bch}{\mathrm{BCH}}
\DeclareMathOperator{\tr}{Tr}
\newcommand{\symdif}{\triangle}

\begin{document}

\begin{frontmatter}

\title{The Second Generalized Covering Radius of Binary Primitive Double-Error-Correcting BCH Codes}
\author[1]{Lev Yohananov\corref{cor1}}
\ead{levyuhananov@gmail.com}
\author[1,2]{Moshe Schwartz}
\ead{schwartz.moshe@mcmaster.ca}
\address[1]{School of Electrical and Computer Engineering, Ben-Gurion University of the Negev, Beer Sheva 8410501, Israel}
\address[2]{Department of Electrical and Computer Engineering, McMaster University, Hamilton, Ontario L8S 4K1, Canada}
\cortext[cor1]{Corresponding author. }

\begin{abstract}
We completely determine the second covering radius for binary primitive double-error-correcting BCH codes. As part of this process, we provide a lower bound on the second covering radius for binary primitive BCH codes correcting more than two errors.
\end{abstract}


\end{frontmatter}

\section{Introduction}

A $q$-ary $[n,k]$ code, $\cC$, is simply a $k$-dimensional subspace of $\F_q^n$, where $\F_q$ denotes the finite field of size $q$. The elements of $\cC$ are called codewords. One common way of describing such a code is by using a parity-check matrix, $H\in\F_q^{(n-k)\times n}$, such that $\bfc\in \cC$ if and only if $H\bfc^\T=\bzero^\T$, namely, $\cC=\ker H$. For a general vector $\bfx\in\F_q^n$, we call $H\bfx^\T$ the syndrome of $\bfx$, and thus the codewords of $\cC$ are exactly the vectors with syndrome $\bzero^\T$.

Since codes are also geometric objects, we need to define a relevant metric. We shall employ the commonly used Hamming metric. For any two vectors, $\bfx=(x_1,\dots,x_n), \bfx'=(x'_1,\dots,x'_n)\in\F_q^n$, the Hamming distance between them is defined by
\[
d(\bfx,\bfx')\eqdef \abs*{\set*{1\leq i\leq n ~:~ x_i\neq x'_i}}.
\]
In this metric, a ball of radius $r$ centered at $\bfx$ is then defined as
\[
B_r(\bfx) \eqdef \set*{ \bfx'\in\F_q^n ~:~ d(\bfx,\bfx')\leq r}.
\]
We conveniently define the minimum distance between a vector $\bfx$ and a set of vectors $\cC$ by
\[
d(\bfx,\cC) \eqdef \min_{\bfc\in\cC} d(\bfx,\bfc).
\]

An important parameter associated with a code, $\cC$, is its covering radius, $R(\cC)$, defined as the largest distance from a point in the space to the code,
\[
R(\cC) \eqdef \max_{\bfx\in\F_q^n}d(\bfx,\cC)=\max_{\bfx\in\F_q^n}\min_{\bfc\in\cC} d(\bfx,\bfc).
\]
Equivalently, the covering radius of $\cC$ is the minimum integer $r$ such that the union of balls of radius $r$ centered at the codewords of $\cC$ is the entire space, i.e., $\bigcup_{\bfc\in\cC}B_r(\bfc) = \F_q^n$. While these two definitions are geometric in nature, yet another equivalent definition for the covering radius is algebraic. It is well known~\cite{CohHonLitLob97}, that $R(\cC)$ is also the minimum integer $r$ such that any syndrome $\bfs^\T\in\F_q^{n-k}$ may be shown as a linear combination of at most $r$ columns of $H$. If we denote the columns of $H$ by $\bfh_1^\T,\dots,\bfh_n^\T$, then $R(\cC)$ is the minimum integer $r$ such that for any $\bfs^\T\in\F_q^{n-k}$ there exist indices $i_1,\dots,i_r$ (that may depend on $\bfs^\T$) such that $\bfs^\T\in\spn{\bfh^\T_{i_1},\dots,\bfh^\T_{i_r}}$, where $\spn{\cdot}$ denotes the linear span. While several parity-check matrices may describe the same code, the covering radius does not depend on the choice of a parity-check matrix (see~\cite{CohHonLitLob97}).

When codes are studied with their covering radius in mind, they are usually called covering codes. These have been extensively studied and have a multitude of applications, ranging from general covering problems, through combinatorial designs, and ending with football pools. In engineering, they are commonly used in compression schemes as well as writing on constrained memories. For all of these and many other results and applications, the reader is referred to~\cite{CohHonLitLob97}.

Recently, in~\cite{EliFirSch21a}, motivated by an application to private information retrieval (PIR) schemes, the notion of the covering radius of a code has been generalized.

\begin{definition}
\label{def:gcr}
Let $\cC$ be an $[n,k]$ $q$-ary code, and let $H\in\F_q^{(n-k)\times n}$ be a parity-check matrix for it whose columns are $\bfh^\T_1,\dots,\bfh^\T_n$. The \emph{$t$-order generalized covering radius of $\cC$}, denoted $R_t(\cC)$, is the minimum integer $r$ such that for any $t$ vectors, $\bfs^\T_1,\dots,\bfs^\T_t\in\F_q^{n-k}$ there exist indices $i_1,\dots,i_r$ (that may depend on $\bfs^\T_1,\dots,\bfs^\T_t$) such that $\set{\bfs^\T_1,\dots,\bfs^\T_t}\subseteq \spn{\bfh_{i_1}^\T,\dots,\bfh_{i_r}^\T}$.
\end{definition}

The $t$-order generalized covering radius is indeed a generalization of the covering radius since $R_1(\cC)=R(\cC)$. This definition creates a generalized covering radius hierarchy, $R_1(\cC),R_2(\cC),\dots$, which mirrors the generalized Hamming weights of codes defined in~\cite{Wei91}. Indeed, connections between the two were noted in~\cite{EliFirSch21a,LanRav23}. Other interesting connections exist with $t$-elevated codes. If $\cC$ is a $q$-ary code with a parity-check matrix $H$, then its $t$-elevated code, $\cC_t$, is the $q^t$-ary code with the same parity-check matrix. These codes were studied in a different context in~\cite{HelKloMyk77,Klo78,Hel79}. Obviously, $\cC_1=\cC$. It was shown in~\cite{EliFirSch21a} that $R_t(\cC_1)=R_1(\cC_t)$. Finally, apart from PIR schemes, generalized covering has found uses in generalized football matches~\cite{EliSch24} and quantized linear computations~\cite{RamRavTam24}.

Considerable effort in the study of covering codes has been devoted to the study of the covering radius of error-correcting codes. Apart from the trivial case of perfect codes (Hamming and Golay codes), we mention as some examples cyclic codes~\cite{Hel85}, MDS codes in general~\cite{GabKlo98,BarGiuPla15}, and Reed-Solomon codes in particular~\cite{Dur94,Ost99}, Melas codes~\cite{ShiHelOzbSol22}, and Zetterberg-type codes~\cite{ShiHelOzb23}, all of these are over $\F_q$. In the binary case, $\F_2$, we mention perhaps the most extensively studied family, Reed-Muller codes \cite{HelKloMyk78,CohLit92,Hou93,Hou97,Hou06,CarMes06}, as well as Preparata codes~\cite{CohKarMatSch85}, and Zetterberg codes~\cite{JinChaLeeCheChe10}. We emphasize that only in rare cases the exact covering radius is found, where in most cases only bounds on it are derived. To date, of these codes, the generalized covering radius hierarchy has been determined for Hamming codes codes~\cite{EliFirSch21a}, repetition codes~\cite{EliWeiSch22}, and bounded for Reed-Muller codes~\cite{EliWeiSch22}.

The focus of this paper is another well-known family of codes, BCH codes, of which binary primitive codes are the most studied, due to their many applications in communication systems.

\begin{definition}
\label{def:bch}
Let $m,e$ be positive integers. The \emph{binary primitive BCH code of length $2^m-1$ and designed distance $2e+1$}, denoted $\bch(e,m)$, is a cyclic code over $\F_2$ with parameters $[n=2^m-1,k\geq 2^m-me-1]$, which is capable of correcting $e$ errors. Let $\alpha\in\F_{2^m}$ be a primitive element. Then $\bch(e,m)$ is defined by the parity-check matrix
\[
H(e,m)\eqdef \begin{pmatrix}
1 & \alpha & \alpha^2 & \dots & \alpha^{n-1} \\
1 & \alpha^3 & \alpha^6 & \dots & \alpha^{3(n-1)} \\
1 & \alpha^5 & \alpha^{10} & \dots & \alpha^{5(n-1)} \\
\vdots & \vdots & \vdots & & \vdots \\
1 & \alpha^{2e-1} & \alpha^{(2e-1)\cdot 2} & \dots & \alpha^{(2e-1)\cdot(n-1)}
\end{pmatrix}.
\]
\end{definition}

We note that in Definition~\ref{def:gcr}, $H(e,m)$ is an $e\times n$ matrix over $\F_{2^m}$. However, by using the well-known isomorphism between $\F_{2^m}$ and $\F_2^m$, we can replace each element of $\F_{2^m}$ by a column vector of length $m$ to obtain a binary parity-check matrix for $\bch(e,m)$ of size $me\times n$. We shall make use of this isomorphism throughout this work, conveniently writing binary vectors and matrices using elements from $\F_{2^m}$.

We further comment that if $2e-1\leq 2^{\ceil{m/2}}$, then $H(e,m)$, when viewed as a binary $me\times n$ matrix, is a full-rank matrix, namely $\bch(e,m)$ is exactly a $[2^m-1,2^m-me-1]$ code (see~\cite{CohHonLitLob97}). We shall therefore often require this condition.

We survey some of the known results on the covering radius of BCH codes. When $e=1$, $\bch(1,m)$ is simply the Hamming code. Since it is perfect, its covering radius is trivially
\[
R_1(\bch(1,m))=1.
\]
The generalized covering radius of $\bch(1,m)$ was determined in~\cite{EliFirSch21a},
\[
R_t(\bch(1,m))=\min\set{t,m},
\]
for all orders $t\geq 1$. For $2$-error-correcting BCH codes, it was shown in~\cite{GorPetZie60} that for all $m\geq 3$,
\begin{equation}
\label{eq:bch1}
R_1(\bch(2,m))=3.
\end{equation}

For $3$-error-correcting BCH codes it took the combined effort of~\cite{HorBer76,AssMat76,Hel78} to prove that for all $m\geq 5$,
\[
R_1(\bch(3,m))=5,
\]
relying heavily on the Carlitz-Uchiyama bound for character sums and a computer search. For general $e$-error-correcting BCH codes, provided $2^m\geq (2e-3)((2e-1)!)^2$, it was shown in~\cite{Coh97} that
\[
R_1(\bch(e,m))=2e-1.
\]
For smaller values of $m$
\[
2e-1\leq R_1(\bch(e,m))\leq 2e,
\]
(see~\cite[Chapter 10.3]{CohHonLitLob97}, and the references therein).

The main contribution of this paper is finding the exact second-order covering radius of binary primitive BCH codes:

\begin{theorem}
\label{th:main}
For all $m\geq 3$,
\[
R_2(\bch(2,m))=\begin{cases}
5 & m\neq 4,\\
6 & m=4.
\end{cases}
\]
\end{theorem}

We comment that the proof of Theorem~\ref{th:main} for the cases $m=5,6$ relies on a computer search.

The paper is organized as follows. In Section~\ref{sec:prem} we present some definitions and relevant known results. We study upper bounds in Section~\ref{sec:upper}, and lower bounds in Section~\ref{sec:lower}. Finally, in Section~\ref{sec:conc} we present the proof for Theorem~\ref{th:main}, and discuss some implications and open questions.

\section{Preliminaries}
\label{sec:prem}

In this section, we bring some notation, definitions, and technical results needed in the following sections.

We use $\F_q$ to denote the finite field of size $q$, and $\F_q^*\eqdef\F_q\setminus\set{0}$. Let $\F_q[x]$ denote the set of polynomials in the indeterminate $x$ with coefficients from $\F_q$, and similarly, $\F_q(x)$ to denote the rational functions in the indeterminate $x$ over $\F_q$. Assume $\bfx=(x_1,\dots,x_n)\in\F_q^n$ is some vector. The support of $\bfx$ is the set of indices containing non-zero entries,
\[
\supp(\bfx)\eqdef\set{i ~:~ x_i\neq 0}.
\]
The Hamming weight of $\bfx$ is its distance from $\bzero$,
\[
\wt(\bfx)\eqdef d(\bfx,\bzero)=\abs*{\supp(\bfx)}.
\]

Two important tools we shall use are characters and the trace function. Since we are interested in binary codes, we shall present these over a field with characteristic $2$. For a more in-depth treatment, the reader is referred to~\cite{LidNie97}. An additive character is a function $\chi:\F_{2^m}\to\C$ such that for all $\beta_1,\beta_2\in\F_{2^m}$,
\[
\chi(\beta_1+\beta_2)=\chi(\beta_1)\chi(\beta_2).
\]
A canonical way of constructing an additive character is by using the trace function, $\tr_{\F_{2^m}/\F_2}:\F_{2^m}\to\F_2$ defined by
\[
\tr_{\F_{2^m}/\F_2}(\beta) \eqdef \beta^{2^0}+\beta^{2^1}+\dots+\beta^{2^{m-1}}.
\]
Since we shall only use $\tr_{\F_{2^m}/\F_2}$ in this paper, we shall simply write $\tr$ for simplicity. It is well known that the trace function is linear, namely, for all $\beta_1,\beta_2\in\F_{2^m}$ we have
\[
\tr(\beta_1+\beta_2)=\tr(\beta_1)+\tr(\beta_2).
\]
One can easily verify that
\begin{equation}
\label{eq:canchar}
\chi(\beta)\eqdef(-1)^{\tr(\beta)},
\end{equation}
is indeed an additive character.

Bounds on character sums, e.g., the Carlitz-Uchiyama bound, are a strong tool (see for example~\cite{LidNie97}). We will make use of a variant, where the summation is over characters of rational functions~\cite[Theorem 1.1]{CocPin06}\footnote{The result in~\cite[Theorem 1.1]{CocPin06} is more general than presented here, using a mix of multiplicative and additive characters, and character extensions. We present a sub-case that is sufficient for our needs.}:

\begin{lemma}[\cite{CocPin06}]
\label{lem:charsum}
Consider the finite field $\F_q$ with characteristic $p$. Let $f(x)\in\F_{q}(x)$ be a rational function,
\[
f(x)=p(x)+\frac{r(x)}{q(x)}, \qquad \deg(r(x))<\deg(q(x)), \quad M\eqdef \deg(p(x)),
\]
where $p(x),r(x),q(x)\in\F_{q}[x]$ are polynomials, and $\gcd(r(x),q(x))=1$. Suppose further that $f(x)$ is non-constant, and that $M=0$ or $p\nmid M$. Write
\[
q(x) = \prod_{i=1}^{Q} q_i(x)^{m_i},
\]
where the $q_i(x)$ are distinct irreducible polynomials over $\F_{q}$, and $m_i\geq 1$. Define
\[
L \eqdef \sum_{i=1}^Q (m_i+1)\deg(q_i(x)).
\]
Then there exist complex numbers $\omega_j\in\C$, $1\leq j\leq M+L-1$, such that
\[
\sum_{\beta\in\F_q\setminus\cS} \chi(f(\beta)) = -\sum_{j=1}^{M+L-1} \omega_j,
\]
where $\cS$ is the set of poles of $f(x)$. Additionally, $\abs{\omega_j}=\sqrt{q}$ for all $j$, except for a single value, $j'$, satisfying $\abs{\omega_{j'}}=1$. Thus,
\[
\abs*{\sum_{\beta\in\F_q\setminus\cS} \chi(f(\beta))} \leq 1+(M+L-2)\sqrt{q}.
\]
\end{lemma}

Another useful result is the following, taken from~\cite[Theorem 6.69 and 6.695]{Ber84}:

\begin{lemma}[\cite{Ber84}]
\label{lem:quadtr}
Consider the polynomial $f(x)=\beta_0+\beta_1 x + \beta_2 x^2\in \F_{2^m}[x]$. Then $f(x)$ has a root in $\F_{2^m}$ if and only if 
\[\tr(\beta_0 \beta_2/\beta_1^2)=0.\]
\end{lemma}

We just observe the obvious, which is that $f(x)$ from Lemma~\ref{lem:quadtr} has a root in $\F_{2^m}$ if and only if it has two roots in $\F_{2^m}$ (not necessarily distinct), since it is quadratic, and its coefficients are from $\F_{2^m}$.

\section{Upper Bound}
\label{sec:upper}

Proving a tight upper bound on $R_2(\bch(2,m))$ requires an elaborate sequence of arguments, so we outline our strategy beforehand. By definition and~\eqref{eq:bch1}, any syndrome of $\bch(2,m)$ may be spanned by (at most) three columns of $H(2,m)$. Thus, given two syndromes, we can definitely span both of them by using at most six columns of $H(2,m)$, at most three for each syndrome. However, even in the worst case in which spanning each syndrome requires exactly three columns, we shall show that we can find two sets of three columns that intersect non-trivially, thus bringing the number of required columns to only five. This shall require solving a set of non-linear equations carefully.

We start with some technical lemmas before reaching the main result of the upper bound.
\begin{lemma}
\label{lem:singlesyndrome}
Assume $\beta_1,\beta_2\in\F_{2^m}$, $\beta_2\neq\beta_1^3$, and consider the following set of equations:
\begin{equation}
\label{eq:singlex}
\begin{split}
x_1+x_2+x_3&=\beta_1,\\
x_1^3+x_2^3+x_3^3&=\beta_2.
\end{split}
\end{equation}
Then this set of equations has a solution over $\F_{2^m}$ if and only if there exists $\theta\in\F_{2^m}^*$ such that $\tr(\theta)=0$ and $\sqrt[3]{(\beta_1^3+\beta_2)/\theta}$ exists. Additionally, the solution satisfies
\[
x_1 = \beta_1+\sqrt[3]{\frac{\beta_1^3+\beta_2}{\theta}}.
\]
\end{lemma}
\begin{proof}
We first observe that by changing $y_i\eqdef x_i+\beta_1$, for all $i\in\set{1,2,3}$, we obtain the equivalent set of equations to~\eqref{eq:singlex}:
\begin{align}
y_1+y_2+y_3&=0, \label{eq:y1}\\
y_1^3+y_2^3+y_3^3&=\beta_1^3+\beta_2. \label{eq:y3}
\end{align}
We further simplify the left-hand side of~\eqref{eq:y3} by using~\eqref{eq:y1} and the characteristic of the field, $2$:
\begin{align*}
y_1^3+y_2^3+y_3^3 &= (y_1+y_2+y_3)^3 + y_1^2 y_2 + y_1^2 y_3 + y_2^2 y_1 + y_2^2 y_3 + y_3^2 y_1 + y_3^2 y_2 \\
&= y_1 y_2 (y_1+y_2) + y_1 y_3 (y_1 + y_3) + y_2 y_3 (y_2+y_3) \\
&= y_1 y_2 y_3 + y_1 y_2 y_3 + y_1 y_2 y_3 \\
& = y_1 y_2 y_3.
\end{align*}
We thus have the following equivalent set of equations to~\eqref{eq:singlex}:
\begin{equation}
\label{eq:singley}
\begin{split}
y_1+y_2+y_3&=0, \\
y_1 y_2 y_3&=\beta_1^3+\beta_2.
\end{split}
\end{equation}

In the first direction, assume $y_i=\lambda_i\in\F_{2^m}$, $i\in\set{1,2,3}$, is a solution to~\eqref{eq:singley}. Notice that $\lambda_1=\lambda_2=\lambda_3=0$ is not a solution since $\beta_1^3+\beta_2\neq 0$ by assumption. W.l.o.g., assume $\lambda_1\neq 0$. Define the polynomial
\[
f(x)\eqdef (x+\lambda_1)(x+\lambda_2)(x+\lambda_3).
\]
Expanding $f(x)$ we get,
\[
f(x) = x^3 + \sigma_2(\lambda_1,\lambda_2,\lambda_3)x+\beta_1^3+\beta_2,
\]
where we used~\eqref{eq:singley} to get $\lambda_1+\lambda_2+\lambda_3=0$, and $\lambda_1\lambda_2\lambda_3=\beta_1^3+\beta_2$, and where $\sigma_2(\lambda_1,\lambda_2,\lambda_3)$ is the second elementry symmetric function,
\[
\sigma_2(\lambda_1,\lambda_2,\lambda_3)\eqdef \lambda_1\lambda_2+\lambda_1\lambda_3+\lambda_2\lambda_3.
\]
But now,
\[
0=f(\lambda_1)=\lambda_1^3+\sigma(\lambda_1,\lambda_2,\lambda_3)\lambda_1+\beta_1^3+\beta_2,
\]
so
\[
\sigma_2(\lambda_1,\lambda_2,\lambda_3)=\lambda_1^2+\frac{\beta_1^3+\beta_2}{\lambda_1}.
\]
Thus,
\[
f(x) = x^3 + \parenv*{\lambda_1^2+\frac{\beta_1^3+\beta_2}{\lambda_1}}x+\beta_1^3+\beta_2=(x+\lambda_1)\parenv*{x^2+\lambda_1 x + \frac{\beta_1^3+\beta_2}{\lambda_1}}.
\]
Define
\[
\theta\eqdef \frac{\beta_1^3+\beta_2}{\lambda_1^3}.
\]
Since $x^2+\lambda_1 x + \frac{\beta_1^3+\beta_2}{\lambda_1}$ needs to have two roots (namely, $\lambda_2$ and $\lambda_3$), by Lemma~\ref{lem:quadtr}, we must have
\[
\tr\parenv*{\frac{\beta_1^3+\beta_2}{\lambda_1^3}}=\tr(\theta)=0,
\]
and obviously, $\sqrt[3]{(\beta_1^3+\beta_2)/\theta}=\lambda_1$ exists (note that $\theta\neq 0$ since $\beta_1^3+\beta_2\neq 0$ by assumption). For the additional claim, the solution to~\eqref{eq:singlex} may be obtained from the solution to~\eqref{eq:singley} by adding $\beta_1$, and thus we get $x_i=\lambda_i+\beta_1$ for all $i\in\set{1,2,3}$, as claimed.

In the other direction, assume $\theta\in\F_{2^m}^*$ such that $\tr(\theta)=0$ and that $\sqrt[3]{(\beta_1^3+\beta_2)/\theta}$ exists. Define
\[
\lambda_1=\sqrt[3]{\frac{\beta_1^3+\beta_2}{\theta}},
\]
and note that $\lambda_1\neq 0$ since $\beta_1^3+\beta_2\neq 0$ by assumption. By Lemma~\ref{lem:quadtr}, the polynomial $x^2+\lambda_1 x + \frac{\beta_1^3+\beta_2}{\lambda_1}$ has two roots, which we shall denote by $\lambda_2,\lambda_3\in\F_{2^m}$. We then observe the polynomial,
\begin{align*}
(x+\lambda_1)(x+\lambda_2)(x+\lambda_3)&=(x+\lambda_1)\parenv*{x^2+\lambda_1 x + \frac{\beta_1^3+\beta_2}{\lambda_1}}\\
&=x^3 + \parenv*{\lambda_1^2+\frac{\beta_1^3+\beta_2}{\lambda_1}}x+\beta_1^3+\beta_2.
\end{align*}
By comparing the coefficient of $x^2$ and the free coefficient on both sides we obtain:
\begin{align*}
\lambda_1+\lambda_2+\lambda_3 &= 0 \\
\lambda_1\lambda_2\lambda_3 &= \beta_1^3+\beta_2,
\end{align*}
thus finding a solution to~\eqref{eq:singley}.
\end{proof}

\begin{lemma}
\label{lem:findL}
For all $m\geq 7$, and for all $\gamma_1,\gamma_2,\delta\in\F_{2^m}$, there exists $\lambda\in\F_{2^m}\setminus\set{0,\delta}$ such that
\begin{equation}
\label{eq:lambdareq}
\tr\parenv*{\frac{\gamma_1}{\lambda^3}}=\tr\parenv*{\frac{\gamma_2}{(\lambda+\delta)^3}}=0.
\end{equation}
\end{lemma}
\begin{proof}
We divide the proof into several cases, and start with the most difficult case.

\emph{Case 1:} $\gamma_1,\gamma_2,\delta\neq 0$. We make use of the canonical additive character $\chi$ given in~\eqref{eq:canchar}. We define the following function:
\[
f(x)=\frac{1+\chi(\gamma_1/x^3)}{2}\cdot\frac{1+\chi(\gamma_2/(x+\delta)^3)}{2}.
\]
We observe that $f(x)=1$ exactly when both $\tr(\gamma_1/x^3)=0$ and $\tr(\gamma_2/(x+\delta)^3)=0$. Otherwise, $f(x)=0$. Thus, the number of $\lambda\in\F_{2^m}\setminus\set{0,\delta}$ that satisfy~\eqref{eq:lambdareq} is
\begin{align*}
\sum_{x\in\F_{2^m}\setminus\set{0,\delta}} f(x) &=
\sum_{x\in\F_{2^m}\setminus\set{0,\delta}}\frac{1}{4}\bigg(1+\chi\parenv*{\frac{\gamma_1}{x^3}}+\chi\parenv*{\frac{\gamma_2}{(x+\delta)^3}}\\
&\qquad\qquad\qquad\qquad +\chi\parenv*{\frac{\gamma_1}{x^3}}\chi\parenv*{\frac{\gamma_2}{(x+\delta)^3}}\bigg)\\
&=
\frac{2^m-2}{4}+ \frac{1}{4}\sum_{x\in\F_{2^m}\setminus\set{0,\delta}}\chi\parenv*{\frac{\gamma_1}{x^3}}+\frac{1}{4}\sum_{x\in\F_{2^m}\setminus\set{0,\delta}}\chi\parenv*{\frac{\gamma_2}{(x+\delta)^3}}\\
&\qquad\qquad+\frac{1}{4}\sum_{x\in\F_{2^m}\setminus\set{0,\delta}}\chi\parenv*{\frac{\gamma_1(x+\delta)^3+\gamma_2 x^3}{(x(x+\delta))^3}}.
\end{align*}
We now use Lemma~\ref{lem:charsum} and get that
\begin{align*}
\abs*{\sum_{x\in\F_{2^m}\setminus\set{0,\delta}}\chi\parenv*{\frac{\gamma_1}{x^3}}} &\leq 1+2\sqrt{2^m},\\
\abs*{\sum_{x\in\F_{2^m}\setminus\set{0,\delta}}\chi\parenv*{\frac{\gamma_2}{(x+\delta)^3}}} &\leq 1+2\sqrt{2^m},\\
\abs*{\sum_{x\in\F_{2^m}\setminus\set{0,\delta}}\chi\parenv*{\frac{\gamma_1(x+\delta)^3+\gamma_2 x^3}{(x(x+\delta))^3}}} &\leq 1+6\sqrt{2^m}.
\end{align*}
Thus,
\[
\sum_{x\in\F_{2^m}\setminus\set{0,\delta}} f(x) \geq \frac{1}{4}\parenv*{2^m-5-10\sqrt{2^m}}>0,
\]
and there exists $\lambda$ as desired, where the last strong inequality follows from $m\geq 7$.

\emph{Case 2:} $\gamma_1,\gamma_2\neq 0$ but $\delta=0$. We proceed as in Case 1, defining the same function $f(x)$, and getting the simpler expression,
\begin{align*}
\sum_{x\in\F_{2^m}^*} f(x) &=
\frac{2^m-1}{4}+ \frac{1}{4}\sum_{x\in\F_{2^m}^*}\chi\parenv*{\frac{\gamma_1}{x^3}}+\frac{1}{4}\sum_{x\in\F_{2^m}^*}\chi\parenv*{\frac{\gamma_2}{x^3}}\\
&\qquad\qquad+\frac{1}{4}\sum_{x\in\F_{2^m}^*}\chi\parenv*{\frac{\gamma_1+\gamma_2}{x^3}}.
\end{align*}
Bounds on the first two character sums are essentially the same as in Case 1. The last summand is either $(2^m-1)/4$ if $\gamma_1=\gamma_2$, and otherwise, by Lemma~\ref{lem:charsum},
\[
\abs*{\sum_{x\in\F_{2^m}^*}\chi\parenv*{\frac{\gamma_1+\gamma_2}{x^3}}} \leq 1+2\sqrt{2^m}.
\]
In any case,
\[
\sum_{x\in\F_{2^m}^*} f(x) \geq \frac{1}{4}\parenv*{2^m-4-6\sqrt{2^m}}>0,
\]
and there exists $\lambda$ as desired, where the last strong inequality follows from $m\geq 7$.

\emph{Case 3:} $\gamma_1=0$ or $\gamma_2=0$. This simply eliminates constraints on the desired $\lambda$, so the set of solutions here is a superset of those of Case 1 if $\delta\neq 0$, or those of Case 2 if $\delta=0$.
\end{proof}

We are now ready to prove the main upper bound on the second covering radius of $\bch(2,m)$.

\begin{theorem}
\label{th:upper}
For all $m\geq 7$, we have
\[
R_2(\bch(2,m)) \leq 5.
\]
\end{theorem}
\begin{proof}
Let $H(2,m)$ be the parity-check matrix for $\bch(2,m)$, as given in Definition~\ref{def:bch}. Denote the columns of $H(2,m)$ by $\bfh_1^\T,\dots,\bfh_{2^m-1}^\T$, and define $H\eqdef\set{\bfh_1^\T,\dots,\bfh_{2^m-1}^\T}$.

Consider two arbitrary syndromes,
\begin{align*}
\bfs_1^\T&=\begin{pmatrix} \beta_1 \\ \beta_2\end{pmatrix}, &
\bfs_2^\T&=\begin{pmatrix} \beta_3 \\ \beta_4\end{pmatrix}, &
\end{align*}
with $\beta_1,\dots,\beta_4\in\F_{2^m}$. Our goal is to show that we can always find a set of no more than five columns of $H(2,m)$ whose span includes $\bfs_1^\T$ and $\bfs_2^\T$.

By the first covering radius of $\bch(2,m)$ given in~\eqref{eq:bch1}, there exist two subsets of columns, $H_1,H_2\subseteq H$, with $\abs{H_1},\abs{H_2}\leq 3$, such that
\begin{align}
\label{eq:colform}
\sum_{\bfh^\T\in H_1} \bfh^\T &=\bfs_1^\T, &
\sum_{\bfh^\T\in H_2} \bfh^\T &=\bfs_2^\T.
\end{align}
There may be more than one way of finding $H_1$ and $H_2$. If we can find $H_i$ such that $\abs{H_i}\leq 2$, then surely
\[
\bfs_1^\T,\bfs_2^\T\in\spn*{H_1\cup H_2} \qquad\text{and}\qquad \abs*{H_1\cup H_2}\leq 5,
\]
so we are done. In particular, we mention the extreme case of $\bfs_i^\T=\bzero^\T$ for which we can take $H_i=\emptyset$. 

We are therefore left with the case where all choices of $H_i$ satisfy $\abs{H_i}=3$. We now show that we can always find $H_1$ and $H_2$ such that $H_1\cap H_2\neq \emptyset$, and therefore $\abs{H_1\cup H_2}\leq 5$. Rewriting~\eqref{eq:colform} in full detail, we are looking for $x_1,\dots,x_6\in\F_{2^m}$ such that
\begin{align}
x_1+x_2+x_3 &= \beta_1, & x_1^3+x_2^3+x_3^3 &= \beta_2,  \label{eq:x123}\\
x_4+x_5+x_6&=\beta_3, & x_4^3+x_5^3+x_6^3&=\beta_4, \label{eq:x456}
\end{align}
and to ensure non-empty intersection of $H_1$ and $H_2$, we would like to have $x_1=x_4$.

By Lemma~\ref{lem:findL}, there exists $\lambda\in\F_{2^m}^*$ such that
\[
\tr\parenv*{\frac{\beta_1^3+\beta_2}{\lambda^3}}=\tr\parenv*{\frac{\beta_3^3+\beta_4}{(\lambda+\beta_1+\beta_3)^3}}=0.
\]
For convenience, let us denote
\begin{align*}
\theta_1&=\frac{\beta_1^3+\beta_2}{\lambda^3}, &
\theta_2&=\frac{\beta_3^3+\beta_4}{(\lambda+\beta_1+\beta_3)^3}.
\end{align*}
The fact that $\abs{H_1}=3$ implies that $\beta_2\neq\beta_1^3$, for otherwise, we could choose $x_1=x_2=x_3=\beta_1$, thus only using one column of $H$, i.e., $H_1=1$. Similarly, $\abs{H_2}=3$ implies that $\beta_4\neq\beta_3^3$. By Lemma~\ref{lem:singlesyndrome}, \eqref{eq:x123} is solvable with
\[
x_1=\beta_1+\sqrt[3]{\frac{\beta_1^3+\beta_2}{\theta_1}}=\lambda+\beta_1.
\]
Similarly, by Lemma~\eqref{lem:singlesyndrome}, \eqref{eq:x456} is solvable with
\[
x_4=\beta_3+\sqrt[3]{\frac{\beta_3^3+\beta_4}{\theta_2}}=(\lambda+\beta_1+\beta_3)+\beta_3=\lambda+\beta_1,
\]
which proves our claim.
\end{proof}

\begin{remark}
\label{rem:m56}
The proof of Lemma~\ref{lem:findL} uses bounds on character sums of rational functions. For small values of $m$ these are not tight enough. However, a simple computer search has verified that the claim presented in Lemma~\ref{lem:findL} also holds for $m=5,6$.
Thus, by relying on this computer search, the upper bound of Theorem~\ref{th:upper} also holds for $m=5,6$.
\end{remark}

\section{Lower Bound}
\label{sec:lower}

We show a general lower bound on $R_2(\bch(e,m))$, which for $e=2$ will be sufficient for the proof of Theorem~\ref{th:main}, except for the case of $m=4$.

\begin{theorem}
\label{th:lower}
For $e\geq 2$ and $m$ such that $2^{\ceil{m/2}}\geq 2e-1$, we have
\[
R_2(\bch(e,m)) \geq 3e-1.
\]
\end{theorem}
\begin{proof}
Assume to the contrary that $R_2(\bch(e,m))\leq 3e-2$. We choose the following two vectors from $\F_{2^m}^{e}$:
\begin{align*}
\bfs^\T_1 & = \begin{pmatrix} 0 \\ \vdots \\ 0 \\ \beta_1 \end{pmatrix}, &
\bfs^\T_2 & = \begin{pmatrix} 0 \\ \vdots \\ 0 \\ \beta_2 \end{pmatrix},
\end{align*}
where $\beta_1,\beta_2\in\F_{2^m}^*$, $\beta_1\neq\beta_2$. Recall that the length of the code $\bch(e,m)$ is $n=2^m-1$. By definition~\ref{def:gcr}, there exist $\bfx_1,\bfx_2\in\F_2^n$ such that
\begin{align*}
H(e,m)\bfx_1^\T &= \bfs_1^\T, & H(e,m)\bfx_2^\T &= \bfs_2^\T,
\end{align*}
where
\begin{equation}
\label{eq:supp}
\abs*{\supp(\bfx_1)\cup\supp(\bfx_2)} \leq 3e-2.
\end{equation}
Obviously, $\bfx_1,\bfx_2\neq \bzero$, as well as $\bfx_1+\bfx_2\neq\bzero$.

By removing the last row of $\bfs_1^\T$, $\bfs_2^\T$, and $H(e,m)$, we obtain
\[
H(e-1,m)\bfx_1^\T=H(e-1,m)\bfx_2^\T=\bzero^\T.
\]
Namely, $\bfx_1$ and $\bfx_2$ are two distinct non-zero codewords of $\bch(e-1,m)$, as is also $\bfx_1+\bfx_2$, by linearity. The minimum distance of $\bch(e-1,m)$ is at least $2e-1$, and so
\begin{equation}
\label{eq:wt}
\wt(\bfx_1),\wt(\bfx_2),\wt(\bfx_1+\bfx_2) \geq 2e-1.
\end{equation}
Then, by~\eqref{eq:supp},
\begin{align*}
\abs*{\supp(\bfx_1)\cap\supp(\bfx_2)} &\geq \wt(\bfx_1)+\wt(\bfx_2)-\abs*{\supp(\bfx_1)\cup\supp(\bfx_2)}\\
& \geq 2(2e-1)-(3e-2) = e.
\end{align*}
But then,
\begin{align*}
\wt(\bfx_1+\bfx_2) &\leq 3e-2-\abs*{\supp(\bfx_1)\cap\supp(\bfx_2)} \leq 2e-2,
\end{align*}
which contradicts~\eqref{eq:wt}.
\end{proof}

For the single case of $m=4$ we have a stronger lower bound.

\begin{theorem}
\label{th:lowerm4}
We have
\[R_2(\bch(2,4))\geq 6.\]
\end{theorem}
\begin{proof}
Let $\alpha\in\F_{2^4}$ be a primitive element, and consider the parity-check matrix $H(2,4)$ given in Definition~\ref{def:bch}. Denote the $i$-th column of $H(2,4)$ by $\bfh_i^\T$, namely,
\[
\bfh_i^\T=\begin{pmatrix} \alpha \\ \alpha^{3i} \end{pmatrix},
\]
and define the set of columns of $H(2,4)$ by $H\eqdef\set{\bfh_1^\T,\dots,\bfh_{15}^\T}$.

We define the following two syndromes:
\begin{align*}
\bfs_1^\T &= \begin{pmatrix} 0 \\ 1 \end{pmatrix}, &
\bfs_2^\T &= \begin{pmatrix} 0 \\ \alpha^3 \end{pmatrix}.
\end{align*}
By Theorem~\ref{th:lower}, we need at least five columns from $H$ to simultaneously span both $\bfs_1^\T$ and $\bfs_2^\T$. We contend that this is insufficient, and at least six columns are required.

Assume to the contrary there exists $H'\subseteq H$, $\abs{H'}=5$, such that $\bfs_1^\T,\bfs_2^\T\in\spn{H'}$. Thus, there exist two subsets, $H_1,H_2\subseteq H'$, $H_1\cup H_2=H'$, such that
\begin{align*}
\sum_{\bfh^\T\in H_1}\bfh^\T &= \bfs_1^\T, &
\sum_{\bfh^\T\in H_2}\bfh^\T &= \bfs_2^\T.
\end{align*}
Our first observation is that $\abs{H_1},\abs{H_2}\geq 3$ since there is no way of getting the $0$  at the top of $\bfs_i^\T$ by summing fewer than three columns of $H$. We cannot have $H_1=H_2=H'$ since $\bfs_1^\T\neq\bfs_2^\T$. We note that 
\[
\sum_{\bfh^\T\in H_1\symdif H_2} \bfh^\T = \sum_{\bfh^\T\in H_1} \bfh^\T + \sum_{\bfh^\T\in H_2} \bfh^\T = \begin{pmatrix} 0 \\ 1+\alpha^3 \end{pmatrix},
\]
where $H_1\symdif H_2$ denotes the symmetric difference of $H_1$ and $H_2$. We now claim that we cannot have $\abs{H_1\cap H_2}\geq 3$ since then $1\leq \abs{H_1\symdif H_2}\leq 2$, and we get a $0$ as the top component when summing fewer than three columns of $H$. 
We are therefore left with only the following cases.

\emph{Case 1:} $\abs{H_1}=\abs{H_2}=3$. Since in this case $\abs{H_1\cap H_2}=1$, there exist distinct $x_1,\dots,x_5\in\F_{2^4}$ such that
\begin{align}
x_1+x_2+x_3 &= 0 & x_1^3+x_2^3+x_3^3 &= 1, \label{eq:m4x123}\\
x_1+x_4+x_5 &= 0 & x_1^3+x_4^3+x_5^3 &= \alpha^3. \label{eq:m4x145}
\end{align}
By applying Lemma~\ref{lem:singlesyndrome} to~\eqref{eq:m4x123} and~\eqref{eq:m4x145}, there must exist $\theta_1,\theta_2\in\F_{2^4}^*$ such that
\begin{equation}
\label{eq:theta12trace}
\tr(\theta_1)=\tr(\theta_2)=0,
\end{equation}
and
\begin{equation}
\label{eq:theta12cubic}
x_1=\sqrt[3]{\frac{1}{\theta_1}}=\sqrt[3]{\frac{\alpha^3}{\theta_2}}.
\end{equation}
This implies that $\theta_1$ has a cubic root. Since in $\F_{2^4}^*$ we have
\begin{align}
\tr(\alpha^0)& =0, &
\tr(\alpha^3)=\tr(\alpha^6)=\tr(\alpha^9)=\tr(\alpha^{12})=1, \label{eq:f16traces}
\end{align}
it follows that $\theta_1=1$, the only element in $\F_{2^4}^*$ that has a cubic root and trace $0$. 

By~\eqref{eq:theta12cubic}, we then have $\theta_2=\alpha^3\neq 1$, so by~\eqref{eq:f16traces}, $\tr(\theta_2)\neq 0$, contradicting~\eqref{eq:theta12trace}.

\emph{Case 2:} $\abs{H_1}=4$ and $\abs{H_2}=3$. Thus, there exist distinct $x_1,\dots,x_5\in\F_{2^4}$ such that
\begin{align}
x_2+x_3+x_4+x_5 &= 0 & x_2^3+x_3^3+x_4^3+x_5^3 &= 1, \label{eq:m4xx2345}\\
x_1+x_2+x_3 &= 0 & x_1^3+x_2^3+x_3^3 &= \alpha^3. \label{eq:m4xx123}
\end{align}
Summing~\eqref{eq:m4xx2345} and~\eqref{eq:m4xx123} we get
\begin{align}
x_1+x_4+x_5 &= 0 & x_1^3+x_4^3+x_5^3 &= \alpha^3+1. \label{eq:m4xx145}
\end{align}
Again, by applying Lemma~\ref{lem:singlesyndrome} to~\eqref{eq:m4xx123} and~\eqref{eq:m4xx145}, there must exist $\theta_1,\theta_2\in\F_{2^4}^*$ such that
\begin{equation}
\label{eq:xtheta12trace}
\tr(\theta_1)=\tr(\theta_2)=0,
\end{equation}
and
\[x_1=\sqrt[3]{\frac{\alpha^3}{\theta_1}}=\sqrt[3]{\frac{\alpha^3+1}{\theta_2}}.\]
This implies that $\theta_1$ must have a cubic root, and by~\eqref{eq:f16traces}, it can only be $\theta_1=1$. We then get
\[
\theta_2=1+\alpha^{-3}.
\]
Observe that $\alpha^{-3}$ has a cubic root, and $\alpha^{-3}\neq 1$, so by~\eqref{eq:f16traces}, $\tr(\alpha^{-3})=1$. But then
\[
\tr(\theta_2)=\tr(1+\alpha^{-3})=\tr(1)+\tr(\alpha^{-3})=1,
\]
contradicting~\eqref{eq:xtheta12trace}.

\emph{Case 3:} $\abs{H_1}=3$ and $\abs{H_2}=4$. Then there exist distinct $x_1,\dots,x_5\in\F_{2^4}$ such that
\begin{align}
x_2+x_3+x_4+x_5 &= 0 & x_2^3+x_3^3+x_4^3+x_5^3 &= \alpha^3, \label{eq:m4xxx2345}\\
x_1+x_2+x_3 &= 0 & x_1^3+x_2^3+x_3^3 &= 1. \label{eq:m4xxx123}
\end{align}
Summing~\eqref{eq:m4xxx2345} and~\eqref{eq:m4xxx123} we get
\begin{align}
x_1+x_4+x_5 &= 0 & x_1^3+x_4^3+x_5^3 &= \alpha^3+1. \label{eq:m4xxx145}
\end{align}
By applying Lemma~\ref{lem:singlesyndrome} to~\eqref{eq:m4xxx123} and~\eqref{eq:m4xxx145}, there must exist $\theta_1,\theta_2\in\F_{2^4}^*$ such that
\begin{equation}
\label{eq:xxtheta12trace}
\tr(\theta_1)=\tr(\theta_2)=0,
\end{equation}
and
\[x_1=\sqrt[3]{\frac{1}{\theta_1}}=\sqrt[3]{\frac{\alpha^3+1}{\theta_2}}.\]
This implies that $\theta_1$ must have a cubic root, and by~\eqref{eq:f16traces}, it can only be $\theta_1=1$. We then get
\[
\theta_2=1+\alpha^{3}.
\]
Observe that $\alpha^{3}$ has a cubic root, and $\alpha^{3}\neq 1$, so by~\eqref{eq:f16traces}, $\tr(\alpha^{3})=1$. But then
\[
\tr(\theta_2)=\tr(1+\alpha^{3})=\tr(1)+\tr(\alpha^{3})=1,
\]
contradicting~\eqref{eq:xxtheta12trace}.
\end{proof}

\section{Conclusion}
\label{sec:conc}

With the combination of upper and lower bounds presented in the previous sections, we can prove the main theorem.

\begin{proof}[Proof of Theorem~\ref{th:main}]
By Theorem~\ref{th:upper} and Theorem~\ref{th:lower} we get the main claim for all $m\geq 7$. We are left with the cases of $m\in\set{3,4,5,6}$. Using a simple computer search, Remark~\ref{rem:m56} extends the main claim to $m=5,6$, namely,
\[
R_2(2,m)=5 \quad\text{for all $m\geq 5$.}
\]

For $m=4$ we have Theorem~\ref{th:lowerm4} giving us $R_2(2,4)\geq 6$. Since~\eqref{eq:bch1} implies any syndrome is spanned by at most three columns of $H(2,m)$, then any two syndromes are spanned by at most six. Hence, $R_2(2,4)\leq 6$, giving us
\[
R_2(2,4)=6.
\]

Finally, for $m=3$, $\bch(2,3)$ is simply the binary repetition code of length $7$. It is known~\cite{EliWeiSch22}, that for a $q$-ary repetition code of length $n$, the order-$t$ covering radius is $n-\ceil{n/q^t}$. Thus, in our case
\[
R_2(2,3)=7-\ceil*{\frac{7}{2^2}}=5,
\]
and the proof is complete.
\end{proof}

This work takes a step towards a complete mapping of the generalized covering radius hierarchy of BCH codes. As an implication of Theorem~\ref{th:main} we mention its affect on higher order covering radii. It was shown in~\cite{EliFirSch21a} that for any code $\cC$, and any positive integers $t_1,t_2$,
\[
R_{t_1+t_2}(\cC)\leq R_{t_1}(\cC)+R_{t_2}(\cC).
\]
This subadditivity, together with Theorem~\ref{th:main}, implies that for all $t\geq 1$, and all $m\geq 5$,
\[
R_t(2,m) \leq \floor*{\frac{t}{2}}\cdot R_2(2,m) + (t \bmod 2)\cdot R_1(2,m) \leq \frac{5t+1}{2}.
\]

Many open questions remain. In particular, we mention extending the results to higher order covering radii for $\bch(2,m)$, and to the triple-error-correcting $\bch(3,m)$.

\bibliographystyle{elsarticle-num} 
\bibliography{allbib}

\end{document}